\documentclass[12pt, letter]{article}
\usepackage[letterpaper, left=1.truein, right=.8truein, top = 1.truein, bottom = 1.truein]{geometry}

\usepackage{eulervm}
\usepackage{tgpagella}
\usepackage[T1]{fontenc}
\usepackage{amsmath}
\usepackage{amssymb}
\usepackage{amsthm}
\usepackage{cite}
\usepackage{color}
\usepackage{graphicx}
\usepackage{color}
\usepackage{bm}
\usepackage{mathtools}

\DeclareMathAlphabet{\mathpzc}{OT1}{pzc}{m}{it}

\usepackage{algorithm}
\usepackage{algpseudocode}
\makeatletter \def\BState{\State\hskip-\ALG@thistlm} \makeatother

\newcommand\be{\begin{equation}}
\newcommand\ee{\end{equation}}

\newtheorem{thm}{Theorem}[section]
\newtheorem{lem}[thm]{Lemma}

\newtheorem{rem}[thm]{Remark}
\newtheorem{obs}[thm]{Observation}

\newcommand\path{\bm{\pi}}

\DeclarePairedDelimiterX\set[1]\lbrace\rbrace{#1}

\newcommand{\R}{\ensuremath{\mathbb{R}}}

\usepackage{algorithm}
\usepackage{mathtools}
\usepackage[shortlabels]{enumitem}
\usepackage{algpseudocode}
\usepackage{hyperref}
\usepackage{tikz}
\usetikzlibrary{backgrounds}
\usetikzlibrary{intersections}
\usetikzlibrary{positioning}
\usepackage{multicol}

\makeatletter \def\BState{\State\hskip-\ALG@thistlm} \makeatother

\usepackage{graphicx}
\usepackage{subcaption}
\usepackage{layouts}
\makeatletter
\def\blfootnote{\xdef\@thefnmark{}\@footnotetext}
\makeatother

\begin{document}


\title{Fair Allocation in Crowd-Sourced Systems}


\author{Mishal Assif P K$^\dag$, William Kennedy$^\ddag$, Iraj Saniee$^\ddag$  }
\date{%
    $^\dag$ECE, University of Illinois, Urbana-Champaign\\%
    $^\ddag$Math \& Algorithms Group, Bell Labs, Nokia\\[2ex]%
    Date: \today \footnotetext{\hspace{0.4 cm}Work done while author ($\dag$) was an intern at Nokia Bell labs.}
}

\maketitle
\abstract{In this paper, we address the problem of fair sharing of the total value of a crowd-sourced network system between major participants (founders) and minor participants (crowd) using cooperative game theory.  Shapley allocation is regarded as a fair way for computing the shares of all participants in a cooperative game when the values of all possible coalitions could be quantified.
We define a class of value functions for crowd-sourced systems which capture the contributions of the founders and the crowd plausibly and derive closed-form expressions for Shapley allocations to both. These value functions are defined for different scenarios, such as presence of oligopolies or geographic spread of the crowd, taking network effects, including Metcalfe’s law, into account. A key result we obtain is that under quite general conditions, the crowd participants are collectively owed a share between $\frac{1}{2}$ to $\frac{2}{3}$ of the total value of the crowd-sourced system.
We close with an empirical analysis demonstrating consistency of our results with the compensation offered to the crowd participants in some public internet content sharing companies. 
 }

\section{Introduction}
Many existing and emerging online network systems and services are designed to work semi-autonomously using cloud hardware and secure software created by a small group of founders and enabled by mass participation of crowds.  We refer to (the small number of) the former as \emph{major} participants and (the large number of) the latter as \emph{minor} participants.  Without the hardware platform, and the associated reliability and security provided by the software, there would be no platform and service. But without participation of 
the crowd there would be no data or content to enable the service. Recent examples of such platforms include Waze\cite{Waze} and Helium\cite{Helium} and one can count Google and Facebook as older and well-established instances of such services. Recently, cryptographically strong decentralized mechanisms are also added for  higher levels of security in these systems \cite{Helium}. There is a large body of literature in the computer science community on the design and operation of such decentralized services (see \cite{dean09} for an overview) and a smaller but growing number of publications on reliable distributed mechanisms that automatically account for and accumulate rewards for participants \cite{sv19}. 

What has not been discussed much in the said literature concerning crowd-sourced network systems and services is the concept of fair allocation of the service's total value to the crowd participants who by constantly feeding data and information to the system are key to its success. One observes a trend towards rewarding crowd participants via automatically-counted tokens based on such measures as the number of queries, messages or packets each participant processes, independently of the total value that the network system and service as a whole generates.  This particularly obscures the fact that in many such network systems crowd participants additionally provide critical local or even private information free of charge purely as part of participation which is monetized by the service provider, adding to the total value generated by the service.  

Lack of valued-based allocation to participants in large-scale crowd-sourced systems is felt keenly by the general public after over two decades since the emergence of these once-novel enterprises (see, for example, \cite{Lan13}) but interestingly, this topic has not been addressed adequately by the research and networking communities; for some notable exceptions with respect to the narrower Internet Service Provider settlements, see \cite{Ma10,MaRu,Cheung08,Mis10}. 

Motivated by the interest in revisiting some of the structural aspects of crowd-sourced systems by the Web 3.0 Foundation, see \cite{Web3}, and more directly, the new possibilities offered by the enterprise-sourced systems as part of the Industry 4.0, this paper aims to help bridge this gap via a formal methodology for fair allocation of the total value and a framework of what fair allocation could potentially mean for crowd-sourced systems' participants.


The paper is organized as follows. In section \ref{sec2} we describe the role and requirements of value functions that are necessary for computation of the Shapley value in a cooperative games.  Our aim is to make these functions as simple as possible while capturing the key contributions of major and minor participants in crowd-sourced systems.
In section \ref{sec3} we consider various models for a single crowd-sourced system consisting of a major participant or founder, and a large number of minor participants, or crowd. We derive closed-formed expressions for Shapley allocations to both for a general class of value functions under different regimes. A key result here is that under quite general and plausible conditions, crowd participants collectively get a payoff at least $\frac{1}{2}$ of the total value of the crowd-sourced system.

We next consider a broad extension of our methodology to oligopolies of crowed-sourced network systems in section \ref{sec4} whereby distinct collections of single-founder crowd-sourced systems agree to cooperate via pair-wise agreements.  We show similar results hold in terms of fair share of participants.  Interestingly and in contrast to the single crowd-sourced system, the ratio of the fair share of major to minor participants now increases with the number of inter-crowd connections each major participant contributes to. In section \ref{sec5}, we present a model of geographic crowd-source systems, and obtain closed-form expressions for the Shapley allocation to each community which continues to exhibit the $\frac{1}{2}$ to $\frac{2}{3}$ allocation to the crowd participants.   We close with an empirical study of crowd-sourced systems whose public financial statements help estimate the actual revenue share of the crowd which we observe to be consistent with our models' predictions.

\section{Fair Allocation in Cooperative Games}
\label{sec2}
A formal derivation of fair allocation in cooperative games was first introduced by Shapley \cite{shapley} and numerous extensions of it have been considered since.   A good introduction is found in \cite{roth}.  To summarize the main concept, consider the setting where a group of participants $N$ are cooperating towards a common goal and generate a value $\nu(N)$ as a result, and this value needs to be distributed among all the participants in a fair way. Suppose that there is a value function $\nu : 2^{N} \rightarrow \mathbb{R^+}$ that takes any subset $S$ of $N$ as input and outputs the value that would have been generated by the coalition formed by only the participants in $S$, as opposed to all of $N$.  Now let us denote by $\phi_i(\nu)$ the payoff that participant $i$ receives as part of the grand coalition, i.e. when $S = N$.

\emph{Fairness.} There are a few natural properties that a fair allocation scheme must satisfy:
\begin{equation}
\label{eq:shap-ax}
\begin{aligned}
    &1. \ \sum_{i \in N} \phi_i(\nu) = \nu(N) \\
    &2. \ \nu(S \cup \{i\}) = \nu(S \cup \{j\})) \ \text{for all } S \implies \phi_i(\nu) = \phi_j(\nu) \\
    &3. \ \nu(S \cup \{i\}) = \nu(S) \ \text{for all } S \implies \phi_i(\nu) = 0 \\
    &4. \ \phi_i(\nu_1 + \nu_2) = \phi_i(\nu_1) + \phi_i(\nu_2)
\end{aligned}
\end{equation}
Property (3) above \emph{null player} states that any player $i$ that adds no value to any coalition $S$ does not receive any payoff in the grand coalition, while property (2) \emph{symmetry} states that any two players $i$ and $j$ that add the same value to every coalition $S$ should receive the same payoff. Participants are rewarded only based on the value they add and no other factors. Property (1) \emph{Efficiency} states that the total value generated $\nu(N)$ is completely distributed among the participants and property (4) \emph{linearity} says that for two independent games, the fair payoff to each participant must be the sum of the fair payoffs in each independent game.

From these conditions Shapley derives the following expression for the unique payoff $\phi_i(\nu)$ to each participant $i$:

 \begin{align}
     \label{eq:shap-val}
    \phi_i(\nu) &= \sum_{S \subseteq  N-\{i\}} \frac{|S|! ~(|N|-|S|-1)!}{|N|!} ~ (\nu(S \cup \{i\}) -\nu(S)).
  \end{align}
Here the expression in brackets on the extreme right hand side is the marginal value that agent $i$ adds to coalition $S$ and this marginal value is averaged out for all possible $2^{|N|-1}$ coalitions. The payoff $\phi_i(\nu)$ in \ref{eq:shap-val} is known as the \emph{Shapley allocation} to participant $i$.

\section{A Single Crowd-Sourced System}
\label{sec3}
We think of a crowd-sourced system ($CSS$ from here on) as a cooperative game in which one entity (the founder) provides the main infrastructure, which could be physical (hardware) or virtual (software), of a value-generating service that can only succeed if a large number of spatio-temporal agents participate to enable local match of supply and demand for the said infrastructure service.  In this setting, consider a collection of participants denoted by $N=\{g, u_1, u_2, ..., u_n\}$ where $g$ is the founder and $u_i, 1\le i \le n$ are the crowd member identities.  There is much symmetry in this setting which we exploit fully in the sequel.   

It will be assumed that the value $\nu(S)$ of any subset $S$ of $N$ is zero unless it includes the founder, and in the case where the founder is in $S$ its value is equal to a function of the size of the crowd. This is a restatement of the 
notion of a $CSS$ (crowd-sourced system) in that it gives much credit to a founder for 
imagining and initiating a distributed enterprise enabled by a large number of uncoordinated participants who constantly feed the system.  This form of the value function also makes it possible to derive closed-form expressions for the Shapley payoff of all participants in various settings, as we observe below. We use a power of the size of the crowd as its value: a linear function corresponds to standard scaling, the quadratic function has been widely discussed in networking literature as Metcalfe's law \ref{rem:metcalfe}, and higher, super quadratic, powers are included for benchmarking purposes.


\subsection*{Notation}
\begin{flushleft}
\begin{tabular}{ll}
    $g$ & The unique major participant (founder) \\
    $u_i$ & $i$-th minor participant (crowd) \\
    $n$ & Number of minor participants  \\
    $N$ & Set containing founder and \\
    & crowd (grand coalition) \\
    $S$ & A coalition of participants, a subset of $N$ \\
    $|S|$ & Cardinality of $S$ \\
    $\nu(S)$ & Value of the coalition $S$ \\
    $\phi_{u_i}(\nu)$ & Shapley payoff to participant $u_i$ \\
    $K_g$ & Cost incurred by founder per participant \\
    $k_u$ & Cost incurred by each participant \\
    $\rho$ & Constant unit in value functions \\
\end{tabular}
\end{flushleft}

\subsection{Coalition value based on revenue}
\label{sec:sm}
\paragraph{Case 1 (identical participants)}
The first situation we consider is one where all the potential participants are identical. Any coalition that does not contain the major participant $g$ is treated as one with no value, while the value of a coalition including $g$ grows as a power of its size. The motivation for this model comes from Metcalfe's law which states that the value of any network grows proportional to the square of its number of participants. This empirical law, originally devised in the context of ethernet networks \cite{MC40}, seems to also hold remarkably well in many modern situations such as cryptocurrency networks \cite{bitcoinMC} and social networks \cite{socialMC}. The model can be formally written as $N = \{g, u_1, ..., u_n\}$, $\nu: 2^{N} \rightarrow \R$ and
\begin{equation}
\label{eq:smmu-model}
\nu(S) := \begin{cases}
          0 \quad &\text{if} \, g \notin S \\
          \rho|S-\{g\}|^k \quad &\text{if} \, g \in S \\
     \end{cases}
\end{equation}
where $\rho$ is some constant. Metcalfe's law is simply a special case of our model with $k=2$. 

Here we list a few key properties of $\nu(S)$ as defined above and explain why we believe this simple form captures the key aspects of a crowd-sourced system valuation.  We observe that:

1) the founder $g$ plays a critical role in the definition of the value function in that any subset which does not include it has value 0. It is not hard to show that this value function $\nu$ defined in \ref{eq:smmu-model} is supermodular, and the Shapley value's derived in the sequel are stable, see the end of section \ref{sec2}.

2) for a any specific crowd participant, there is no special incentive to join any of the numerous coalitions of size $S$, as the value function only depends on the size of the coalition and not its constitution.  This means the basic assumption in (2) that all coalitions of size $S$ are equivalent is met.

3) the simple form of $\nu$ we assume, linear, quadratic or higher powers, does not need a saturation effect as the more crowd participants, the larger the total value of a coalition.  This is in contrast with saturating systems where (say) the first 100 crowd participants provide higher value than the last 100.

4) we note that in the analysis in the following pages relatively small values of $S$ give us the asymptotic results with respect to the exponent $k$ equals 1 and 2 corresponding to linear and quadratic value functions.

Our first result characterizes the Shapley value associated with each participant as a fraction of the value of the total coalition. We now look at a general computation that will be useful later as well.

\begin{lem}
\label{lem:first}
Consider the coalition game with the set of agents $N = \{g, u_1, ..., u_n\}$ and value $\nu$ given by
\begin{equation*}
\nu(S) := \begin{cases}
          0 \quad &\text{if} \, g \notin S \\
          f\left(|S|\right) \quad &\text{if} \, g \in S \\
     \end{cases}
\end{equation*}
The Shapley values of the game satisfy
\begin{equation*}
    \phi_g(\nu) = \frac{1}{n+1} \sum_{s=0}^n f(s), \quad \phi_{u_i}(\nu) = \frac{\nu(N) - \phi_g(\nu)}{n}.
\end{equation*}
\end{lem}
\begin{proof}
    See Appendix \ref{app3}.
\end{proof}

\begin{thm}
\label{th:smmu-res}
Consider the coalition game with the set of agents $N = \{g, u_1, ..., u_n\}$ and value $\nu$ given by \eqref{eq:smmu-model}. The associated Shapley values satisfy

\begin{equation}
\label{eq:smmu-res1}
     \phi_g(\nu) = \frac{\rho n^k}{k+1} + O(n^{k-1}), \quad \phi_{u_i}(\nu) = \frac{k \rho n^{k-1}}{k+1} - O(n^{k-2})
\end{equation}
and
\begin{equation}
\label{eq:smmu-res}
    \lim_{n \to \infty} \frac{\phi_g(\nu)}{\nu(N)} = \frac{1}{k+1}, \quad \lim_{n \to \infty} \frac{\sum_{i=1}^n \phi_{u_i}(\nu)}{\nu(N)} = \frac{k}{k+1}.
\end{equation}
\end{thm}
\begin{proof}
Applying the results of Lemma \ref{lem:first} to the value \eqref{eq:smmu-model}, we get
\begin{align*}
    \phi_g(\nu) &= \frac{\rho}{n+1} \sum_{s=0}^n s^k =  \frac{1}{n+1}\frac{ \rho n^{k+1} + O\left(n^{k}\right)}{k+1} \\& \approx \frac{\rho n^k}{k+1} = \frac{1}{k+1} \nu(N) \quad \text{for large $n$,} \\
    n\phi_{u_i}(\nu) &= \nu(N) - \phi_g(\nu) \approx \frac{\rho k}{k+1} \nu(N).
\end{align*}
\end{proof}

\begin{rem}
\label{rem:metcalfe}
    \emph{The results obtained in Theorem \ref{th:smmu-res} tell us that when $k = 1$, i.e. when the power of a network grows proportional to the number of participants, the value of the grand coalition is shared equally between the major participant and the collection of minor participants. In other words,}
    $\phi_g(\nu) \approx \frac{\nu(N)}{2}$ and $\sum_{i=1}^n\phi_{u_i}(\nu) \approx \frac{\nu(N)}{2}$.

    \emph{In the popular case of Metcalfe's law where $k=2$, the theorem says that the major participant should get 1/3 of the value generated while 2/3 should be distributed equally among the minor participants.}
\end{rem}

\paragraph{Case 2 (Non-identical participants)}The assumption that all participants are (nearly) identical might not be hold in many settings.  For example, when some participants contribute significantly more to the value of the system than others due to their influence in a social network. A slightly more refined model for the value would be 

\begin{equation}
\label{eq:smmuu-model}
\nu(S) := \begin{cases}
          0 \quad &\text{if} \, g \notin S \\
          \rho \left( \sum_{i \in S} W_i^{\alpha} \right)^k \quad &\text{if} \, g \in S \\
     \end{cases}
\end{equation}
where $W_i$ refers to some notion of amount of work done by participant $i$, e.g., the number of messages routed by $i$. We will denote by $f_i := \frac{W_i^{\alpha}}{\sum_{j} W_j^{\alpha}}$ the share of total work done by $i$.

\begin{thm}
\label{th:smmuu-res}
Consider the coalition game with the set of agents $N = \{g, u_1, ..., u_n\}$ and value $\nu$ given by \eqref{eq:smmuu-model} where $k=2$. The associated Shapley values satisfy
\begin{equation}
\label{eq:smmuu-res}
   \lim_{n \to \infty} \frac{\phi_g(\nu)}{\nu(N)} = \frac{1}{3} + \sum_i \frac{f_i^2}{6} , \quad \lim_{n \to \infty} \frac{ \phi_{u_i}(\nu)}{\nu(N)} = \frac{2}{3}f_i - \frac{f_i^2}{6}.
\end{equation}
\end{thm}
\begin{proof}
    See Appendix \ref{app3}.
\end{proof}

\begin{rem}
\label{rem:alpha-fair}
    \emph{The results obtained in Theorem \ref{th:smmu-res} under the identical participant model can be recovered from the above theorem by assuming $W_i = 1$. In fact, in the case where the workload is more or less uniformly distributed and not too concentrated, as formalized by the condition 
    \[
        \lim_{n \to \infty} \sum_i f_i^2 = 0,
    \]
    the above theorem tells us that $\frac{\phi_g(\nu)}{\nu(N)} \approx \frac{1}{3}$ and $\frac{\phi_{u_i}(\nu)}{\nu(N)} \approx \frac{2}{3}f_i$, which is equivalent to dividing the value between the founder and the crowd in the same way as in the identical participant model, and then dividing the value among the crowd proportional to the share of work done.}
\end{rem}

\begin{rem}
\emph{In theorem \ref{th:smmuu-res}, we have treated the case $k = 2$. However, analogous results hold true for higher values of $k$ as well. It can be shown that}
\begin{align*}
    \lim_{n \to \infty} \frac{\phi_g(\nu)}{\nu(N)} &= \frac{1}{k+1} + O\left(\sum_i f_i^2\right) \\ \lim_{n \to \infty} \frac{ \phi_{u_i}(\nu)}{\nu(N)} &= \frac{k}{k+1}f_i - O\left(f_i^2\right).
\end{align*}
\end{rem}

\subsection{Coalition value based on profit}
A legitimate concern regarding the use of revenue as a proxy for value is that the founder could be incurring significant costs for infrastructure which the minor participants do not. In this case the resulting profit, revenue minus cost, is a more natural proxy for value.  A question which our Shapley allocation derivation answers in the positive is the observable impact of the large infrastructure cost incurred by the founder compared to the small cost incurred by the crowd members. 

We modify the model to take into account such costs. We assume that each minor participant incurs a fixed constant cost $k_u$ and the major participant pays a cost of $K_g$ per minor participant with $K_g >> K_u$. Formalizing all this, we get a new value function $\nu: 2^{N} \rightarrow \R$ 
\begin{equation}
\label{eq:smmu-model-profit}
\nu(S) := \begin{cases}
          0 \quad &\text{if} \, g \notin S \\
             \rho |S-\{g\}|^k - K_g|S-\{g\}| - \\
             \quad k_u|S-\{g\}| \quad &\text{if} \, g \in S \\
     \end{cases}
\end{equation}

We compute the Shapley values in this case by appealing to Lemma \ref{lem:first} again.

\begin{thm}
\label{th:smmu-res-profit}
Consider the coalition game with the set of agents $N = \{g, u_1, ..., u_n\}$ and value $\nu$ given by \eqref{eq:smmu-model-profit}. The associated Shapley values satisfy
\begin{equation}
\begin{aligned}
\label{eq:smmu-res2}
   &\frac{\phi_g(\nu)}{\nu(N)} \approx \frac{\frac{\rho n^{k}}{k+1}-K_g\frac{n}{2}-k_u\frac{n}{2}}{\rho n^k - K_gn - k_un} \\ &\frac{\sum_{i=1}^n \phi_{u_i}(\nu)}{\nu(N)} \approx \frac{\frac{\rho kn^{k}}{k+1}-K_g\frac{n}{2}-k_u\frac{n}{2}}{\rho n^k - K_gn - k_un}
\end{aligned}
\end{equation}
for large values of $n$.
\end{thm}
\begin{proof}

Using the results of Lemma \ref{lem:first}, the exact same computation in the proof of Theorem \ref{th:smmu-res} delivers this result.
\end{proof}

\begin{rem}
    \emph{Even though in the limit $n \to \infty$, the results of Theorem \ref{th:smmu-res-profit} collapse to that of \ref{th:smmu-res}, the value of $K_g$ might be large enough that the contribution of $K_gn$ cannot be ignored for the moderately large values of $n$ that typically arise.}
\end{rem}
\begin{rem}
    \emph{In this case, the ratio of the share of the founder to the share of the crowd is}
    \[
    \frac{\phi_g(\nu)}{\sum_{i =1}^n \phi_{u_i}(\nu)} = \frac{\frac{\rho n^{k}}{k+1}-K_g\frac{n}{2}-k_u\frac{n}{2}}{\frac{\rho kn^{k}}{k+1}-K_g\frac{n}{2}-k_u\frac{n}{2}}.
    \]
    \end{rem}
    This shows that when $k > 1$, as the cost of the system infrastructure $K_g$ increases, the share of the founder compared to the crowd of participants decrease. This may seem at first counterintuitive or unfair, especially compared to the case $k=1$ where the share is constant regardless of the cost of the infrastructure. This phenomenon occurs because the revenue should be divided according to $\frac{1}{k+1}$ and $\frac{k}{k+1}$ between the founder and the crowd, whereas the total cost is split equally.  Thus superlinear scaling of revenue and linear scaling of cost give a decreasing proportional share to the founder even though the absolute value of the founder share is still increasing.
    
    Another observation regarding (6) is that while the share of the founder $\phi_g(\nu)$ is increasing in $k$, the founder's profit decreases as a percentage of the total profit. This again is a consequence of the revenue of the network growing as a power of its size. 

\section{Oligopolies of Crowd-Sourced Systems}   
\label{sec4}

We now consider the case of multiple crowd-sourced systems where major participants or founders are willing to cooperate, thus inter-working $CSS$s for a higher total payoff to all. Let us assume that each major participant $v$ has $n_v$ minor participants in its $CSS(v)$, crowd-sourced system affiliated with $v$. The major participants now agree to cooperate for higher payoff by inter-networking their disparate crowds. This is done through a series of bilateral agreements between pairs of major players.  The key consequence of such cooperation between major participants $i$ and $j$ is joint creation of a payoff which is proportional to $(n_i+n_j)^2$ for the Metcalfe value model we considered previously, in contrast to $n_i^2+n_j^2$ for the case $i$ and $j$ do not cooperate. The question we wish to answer is fair allocation to each $CSS$, and each major and minor participants in each $CSS$ assuming global cooperation represented by a graph $G$ as defined below. 

\subsection*{Additional Notation}

\begin{flushleft}
\begin{tabular}{ll}
    $V$ & Set of vertices $v$, each representing a \\
        & crowd-sourced system, $v \equiv CSS(v)$\\
    $E$ & Set of edges, each edge denoting an agreement \\
        & between its ends \\
    $G$ & Oligopoly graph of $CSS$s and their bilateral\\
        & agreements \\
    $S$ & A subset of vertices $V$ \\
    $G(S)$ & Subgraph of $G$ induced by $S$ \\
    \end{tabular}
\end{flushleft}
\begin{flushleft}
\begin{tabular}{ll}
    $v$ & A vertex/participant in $V$\\
    $e$ & An edge/agreement in $E$\\
    $N(v)$ & Set of neighbors of the vertex $v$ in graph $G$ \\
    $n_v$ & Number of participants(crowd) in $v$'s \\
    & crowd-sourced system \\
\end{tabular}
\end{flushleft}

\subsection{Shapley value of an Oligopoly: Coarse-grain model}
We now assume that major participants have formed a set of pairwise, bilateral, agreements amongst themselves. We represent this by a connected graph $G = (V, E)$, whose vertices $v \in V$ denote $CSS$s associated with major participants and each edge in $E$ denotes a bilateral agreement. A natural extension of the notion of Metcalfe value to a graph $G$ is given by
\begin{equation}
    \label{eq:network-val}
        \nu(G) = \sum_{v \in V} n_v^2 + \sum_{(u, w) \in E} n_un_w.
\end{equation}
Here $n_v^2$ is the Metcalfe law restricted to the system $v$ and each edge $(u, w)$ adds an additional $n_un_w$ to the value of $G$ due to inter-networking  between vertices $u$ and $w$. 

We assume that the crowd-sourced systems associated with major participants $v \in V$ form the set of agents of a cooperative game. The value of a subset $S \subset V$ is the Metcalfe value \eqref{eq:network-val} of the subgraph $G(S)$ induced by $S$. This is the graph consisting of vertices $S$ and all edges of $G$ both ends of which lie in $S$. This is implicitly an aggregate model that clubs together all the minor participants of a $CSS$ with the major participant, modelling the scenario where each minor participant has already been bound to its major participant, perhaps by means of an agreement or due to some external constraints such as geographic that limits the choice of the minor participant. As a result, in the following $\phi_u(\nu)$ denotes the share of the crowd-sourced system associated with major participant $u$, and not just the share of the major participant $u$ in $CSS(u)$ by itself (which will be discussed in Remark 4.4). 

The following result characterizes the Shapley value of each $CSS$.

\begin{thm}
\label{th:network-shap-res}
Consider a graph $G = (V, E)$ and a cooperative game with set of agents $N = V$ and value function
\[
    \nu(S) = \sum_{v \in S} n_v^2 + \sum_{(u, w) \in S\times S \cap E} n_un_w.
\]
The Shapley value of each vertex $u$ is given by
\begin{equation}
    \label{eq:network-shap-res}
    \phi_u(\nu) (\equiv \phi_{CSS(u)}(\nu))= n_u^2 + \sum_{w \in N(u)} n_un_w
\end{equation}
\end{thm}
\begin{proof}
    See Appendix \ref{app4}.
\end{proof}

\begin{figure}
\label{fig:oligo}
\begin{center}
  \begin{tikzpicture}
    \node[shape=circle,draw=black] (A) at (0,3) {$n_A$};
    \node[shape=circle,draw=black] (B) at (2.5,4) {$n_B$};
    \node[shape=circle,draw=black] (C) at (2.5,1) {$n_C$};
    \node[shape=circle,draw=black] (D) at (5,3) {$n_D$} ;

    \path [-](A) edge node[left] {} (B);
    \path [-](C) edge node[left] {} (B);
    \path [-](C) edge node[] {} (A);
    \path [-](B) edge node[] {} (D);
\end{tikzpicture}
\end{center}
\caption{Graph of four crowd-sourced systems with four bilateral agreements.}
\end{figure}
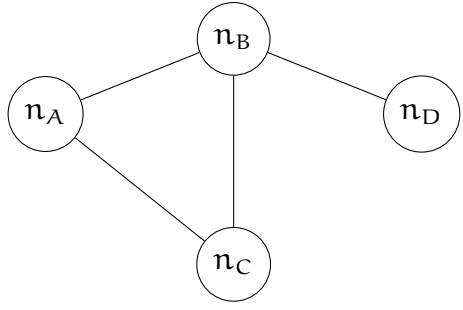

For example, in the oligopoly shown in Figure \ref{fig:oligo}, we can compute the Shapley values as
\begin{align*}
    \phi_{A}(\nu) &= n_A(n_A+n_B+n_C), \\
    \phi_{B}(\nu) &= n_B(n_A+n_B+n_C+n_D), \\
    \phi_{C}(\nu) &= n_C(n_A+n_B+n_C), \\
    \phi_{D}(\nu) &= n_D(n_B+n_D).
\end{align*}

\begin{rem}
\emph{In the special case where all major participants have equal number of minor participants within their system, i.e. $n_v = n$ for all $v \in V$, we get}
\begin{equation}
    \phi_u(\nu) = n^2(1 + \text{\emph{deg}}(u)).
\end{equation}
\end{rem}
\begin{rem}
    \emph{If we consider $V = \{g, u_1, ..., u_n\}$, that there is an edge between $g$ and $u_i$ for each $i = 1,...,n$, and that $n_{u_i} = 1$ and $n_{u_g} = 1$, we get a model that is very close to \eqref{eq:smmu-model} with $k = 2$. The only difference here is that the coalitions that don't contain $g$ have non-zero value that scales linearly with its size. Unsurprisingly, we get}
    \begin{align*}
        &\phi_g(\nu) = 1 + \text{\emph{deg}}(g) = 1 + n, \\ &\phi_{u_i}(\nu) = 2 \ \implies \nu(N) = 3n + 1, \\
        &\implies \lim_{n \to \infty} \frac{\phi_g(\nu)}{\nu(N)} = \frac{1}{3}, \quad \lim_{n \to \infty} \frac{\sum_{i=1}^{n}\phi_{u_i}(\nu)}{\nu(N)} = \frac{2}{3}
    \end{align*}
    \emph{exactly the same as in Theorem \ref{th:smmu-res}.}
\end{rem}

\subsection{Shapley value of an Oligopoly: Fine-grain model}
Theorem \ref{th:network-shap-res} and the Remarks 4.2 and 4.3 derive the Shapely value of each crowd-sourced system $($CSS$)$ associated with each vertex in the oligopoly game of graph $G$.  
What remains is to identify the Shapley fair share of each participant in each $CSS$ associated with each major player $v \in G$.  To do this, we need to define a new value function on each $CSS$ which takes its size and sizes of its neighboring $CSS$s explicitly into account.

To this end, we consider a more fine-grained model here. We consider the set $U_v$ of minor participants in the $CSS$ associated with major participant $v \in V$ as separate agents, making the total set of agents of the cooperative game \[N = \cup_{v \in V} U_v \cup V.\]
A subset $S \subset N$ can be partitioned as $S_v = S \cap U_v$ for each $v \in V$ and $S_V = S \cap V$. The value of the coalition is computed based on two simple principles: each $v \in S_V$ contributes $|S_v|^2$ and each $e = (v,w) \in G(S_V)$ contributes $|S_v||S_w|$ to the value of $S$. That is, 
\begin{equation}
    \label{eq:oli-fine-model}
    \xi(S) = \rho \left(\sum_{v \in S_V} |S_v|^2 + \sum_{(u,w) \in S_V \times S_V \cap E} |S_u||S_w|\right)
\end{equation}
This value function and agent set jointly capture the cooperative game involving all $CSS$s and their minor and major participants. For $v \in V$, we denote by $n_v := |U_v|$ the size of its set of minor participants.

\begin{thm}
\label{th:network-shap-res1}
Consider a graph $G = (V, E)$ where each node $v \in V$ has an associated set $U_v$, and a cooperative game with set of agents $N = \cup_{v \in V} U_v \cup V$ and value function
\[
    \xi(S) = \rho \left(\sum_{v \in S_V} |S_v|^2 + \sum_{(u,w) \in S_V \times S_V \cap E} |S_u||S_w|\right)
\]
where $S_V = S\cap V$ and $S_v = S\cap U_v$ for each $v \in V$. The Shapley value of each major participant $v \in V$ is the sum of two terms, the intra and inter crowd payoffs, given by
\begin{equation}
    \label{eq:network-shap-res2}
    \phi_v(\xi) = \rho \left( \frac{n_v^2}{3} - O(n_v) \right) + \rho \left( \sum_{w \in N(v)} \frac{n_vn_w}{2} \right)
\end{equation}

The Shapley value of each minor participant $u_v \in U_v$ is also the sum of two terms 
\begin{equation}
    \label{eq:network-shap-res-minor}
   \phi_{u_v}(\xi) = \rho \left( \frac{2n_v}{3} + O(1) \right)  + \rho \left(\sum_{w \in N(v)} \frac{n_w}{2}\right).
\end{equation}  
\end{thm}
\begin{proof}
    See Appendix \ref{app4}.
\end{proof}

\begin{rem}
    \emph{The Shapley share of the major participant $u$ in $CSS(u)$ (vertex $u \in G$ of the oligopoly graph $G$) can be written as}
    \begin{align*}
        \frac{\phi_{v} (\xi)}{\phi_v(\xi) + \sum_{u_v \in U_v} \phi_{u_v}(\xi)} &= \frac{ \frac{n_v^2}{3} + O(n_v) + \sum_{w \in N(v)}\frac{n_vn_w}{2}}{ n_v^2 + \sum_{w \in N(v)}n_vn_w} 
    \end{align*}
    \emph{For large enough numbers of minor participants $n_v$, this ratio can be simplified as}
     \begin{align*}
        &\frac{\phi_{v} (\xi)}{\phi_v(\xi) + \sum_{u_v \in U_v} \phi_{u_v}(\xi)} = \frac{ \frac{n_v}{3} + \sum_{w \in N(v)}\frac{n_w}{2}}{ n_v + \sum_{w \in N(v)}n_w} \\
        &= \frac{1}{3}\times \frac{n_v}{n_v + \sum_{w \in N(v)}n_w} + \frac{1}{2}\times \frac{\sum_{w \in N(v)}n_w}{n_v + \sum_{w \in N(v)}n_w}.
    \end{align*}
    \emph{This ratio is greater than $\frac{1}{3}$ and less than $\frac{1}{2}$, moving closer to $\frac{1}{2}$ as the number of minor participants in the neighbouring $CSS$ grows larger than that in the $CSS$ corresponding to $v$.}
\end{rem}

\section{Geographic Crowd-Sourced Systems}

\label{sec5}
In this last section, we discuss the extension of the models we have considered so far by considering $CSS$s which are geographically distributed. More concretely, consider a collection of agents and for each agent the region covered by it. For simplicity we refer to these regions as ``disks", even though these can be of arbitrary shape. Figure \ref{fig:geo-ex} shows five disks associated with five agents.

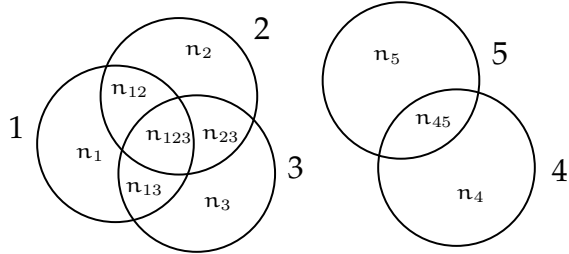
\begin{figure}[!h]
\advance\leftskip-4cm
\centering

\tikzset{every picture/.style={line width=0.75pt}} 

\tikzset{every picture/.style={line width=0.75pt}} 

\begin{tikzpicture}[x=0.75pt,y=0.75pt,yscale=-1,xscale=1]

\draw   (85,87.5) .. controls (85,65.68) and (102.68,48) .. (124.5,48) .. controls (146.32,48) and (164,65.68) .. (164,87.5) .. controls (164,109.32) and (146.32,127) .. (124.5,127) .. controls (102.68,127) and (85,109.32) .. (85,87.5) -- cycle ;
\draw   (53,111.5) .. controls (53,89.68) and (70.68,72) .. (92.5,72) .. controls (114.32,72) and (132,89.68) .. (132,111.5) .. controls (132,133.32) and (114.32,151) .. (92.5,151) .. controls (70.68,151) and (53,133.32) .. (53,111.5) -- cycle ;
\draw   (94,126.5) .. controls (94,104.68) and (111.68,87) .. (133.5,87) .. controls (155.32,87) and (173,104.68) .. (173,126.5) .. controls (173,148.32) and (155.32,166) .. (133.5,166) .. controls (111.68,166) and (94,148.32) .. (94,126.5) -- cycle ;
\draw   (197,79.5) .. controls (197,57.68) and (214.68,40) .. (236.5,40) .. controls (258.32,40) and (276,57.68) .. (276,79.5) .. controls (276,101.32) and (258.32,119) .. (236.5,119) .. controls (214.68,119) and (197,101.32) .. (197,79.5) -- cycle ;
\draw   (225,123.5) .. controls (225,101.68) and (242.68,84) .. (264.5,84) .. controls (286.32,84) and (304,101.68) .. (304,123.5) .. controls (304,145.32) and (286.32,163) .. (264.5,163) .. controls (242.68,163) and (225,145.32) .. (225,123.5) -- cycle ;

\draw (71,112) node [anchor=north west][inner sep=0.75pt]  [font=\scriptsize] [align=left] {$\displaystyle n_{1}$};
\draw (126,60) node [anchor=north west][inner sep=0.75pt]  [font=\scriptsize] [align=left] {$\displaystyle n_{2}$};
\draw (88.5,80) node [anchor=north west][inner sep=0.75pt]  [font=\scriptsize] [align=left] {$\displaystyle n_{12}$};
\draw (106,103) node [anchor=north west][inner sep=0.75pt]  [font=\scriptsize] [align=left] {$\displaystyle n_{123}$};
\draw (96,129.5) node [anchor=north west][inner sep=0.75pt]  [font=\scriptsize] [align=left] {$\displaystyle n_{13}$};
\draw (135,101) node [anchor=north west][inner sep=0.75pt]   [align=left] {$ $};
\draw (135.5,137) node [anchor=north west][inner sep=0.75pt]  [font=\scriptsize] [align=left] {$\displaystyle n_{3}$};
\draw (221,63) node [anchor=north west][inner sep=0.75pt]  [font=\scriptsize] [align=left] {$\displaystyle n_{5}$};
\draw (263,133) node [anchor=north west][inner sep=0.75pt]  [font=\scriptsize] [align=left] {$\displaystyle n_{4}$};
\draw (242,95) node [anchor=north west][inner sep=0.75pt]  [font=\scriptsize] [align=left] {$\displaystyle n_{45}$};
\draw (134.5,102.5) node [anchor=north west][inner sep=0.75pt]  [font=\scriptsize] [align=left] {$\displaystyle n_{23}$};
\draw (37,96) node [anchor=north west][inner sep=0.75pt]   [align=left] {1};
\draw (281,59) node [anchor=north west][inner sep=0.75pt]   [align=left] {5};
\draw (161,46) node [anchor=north west][inner sep=0.75pt]   [align=left] {2};
\draw (178,116.5) node [anchor=north west][inner sep=0.75pt]   [align=left] {3};
\draw (311,118) node [anchor=north west][inner sep=0.75pt]   [align=left] {4};

\end{tikzpicture}
\caption{Geographic $CSS$ with five disks.}
\label{fig:geo-ex}
\end{figure}

In this model there are $m$ agents $M = \{1,2,...,m\}$, and each agent $i$ has a corresponding disk $D_i$. There are also users, set $U$, served by the agents in $M$. A user $u$ may be served by an agent $i$ only if $u$ is in the region $D_i$. Notice that a user in $D_i$ may also be in $D_j$ for one or more $j \neq i$. We need notation for number of users in the intersection of disks. For each subset of agents $S \subset M$, we will denote by $d_S$ the number of users in the geographic area exclusively covered by all $D_i$ for $i \in S$. For example, in Figure \ref{fig:geo-ex} $d_{\{1,2\}}$ is the number of users in $D_1 \cap D_2$ that don't belong to any of the other disks, and $d_{\{1,2,3\}}$ is the number of users in $D_1\cap D_2 \cap D_3$, etc. More formally, $d_S$ is the number of users in $\cap_{i \in S} D_i \cap_{j \in M-S} D_j^c$.

In this setting, we define the value functions $\nu_{lin}$ and $\nu_{Met}$ as follows. For any $S \subset M$, 
\begin{equation}
\label{eq:geo-linmet}
\begin{aligned}
    \nu_{lin}(\{i\}) &= \rho \sum_{i \in S \subset M}  \frac{d_S}{|S|}, \quad 
    \nu_{lin}(S) = \rho \sum_{i \in S} \nu_{lin}(\{i\}), \\
    \nu_{Met}(S) &=  \left( \nu_{lin}(S) \right)^2/\rho.
\end{aligned}
\end{equation}
For example, in the configuration of Figure \ref{fig:geo-ex} we have
\begin{align*}
\nu_{lin}(\{1\}) &= \rho \left(d_1 + \frac{d_{12}+d_{13}}{2} + \frac{d_{123}}{3}\right), \\ \nu_{lin}(\{2,3\}) &= \rho \left(d_2 + d_3 + \frac{d_{12} + d_{13}}{2} + d_{23} + \frac{2d_{123}}{3}\right),
\end{align*}
etc.
and 
\begin{align*}
\nu_{Met}(\{1\}) &= \rho \left(d_1 + \frac{d_{12}+d_{13}}{2} + \frac{d_{123}}{3}\right)^2, 
\\ \nu_{Met}(\{2,3\}) &= \rho \left(d_2 + d_3 + \frac{d_{12} + d_{13}}{2} + d_{23} + \frac{2d_{123}}{3}\right)^2,
\end{align*}
etc.

With these definitions, the geographic $CSS$ model is equivalent to an oligopoly of $CSS$s where the underlying graph $G$ is complete and $n_i := \nu_{lin}(\{i\})$. The Shapley value for each agent under the Metcalfe model is therefore given by Theorem \ref{th:network-shap-res} while the linear model is an additive cooperative game where Shapley value coincides with the value function, namely
\[
    \phi_i(\nu_{lin}) = \nu_{lin}(i)=\rho n_i,
    \quad \phi_i(\nu_{Met}) = \rho n_i\left( \sum_{i=1}^m n_i\right).  
\]

\subsection{Geographic Crowd-Sourced Systems with a Founder}

As a conclusion to the discussion of geographic crowd-sourced systems, we consider the contribution of a founder who establishes the infrastructure for the collection of geographic crowd-sourced system. For example, the agent who provides the software for connecting access points in a distributed wireless network e.g. \cite{Helium}. Just to set the notation, let

\subsection*{Notation}
\begin{flushleft}
\begin{tabular}{ll}
    $g$ & Main founder of the geographic $CSS$s\\
    $M=\{1,..,m\}$ & Set of agents $i$, each representing disk $D_i$\\
    $n_i$ & Number of users in  disk $D_i$.
\end{tabular}
\end{flushleft}

As in section \ref{sec:sm}, we define a macro value function $\xi$ which quantifies the cooperative game involving the founder and any subset of agents $M\cup\{g\}$.
\begin{equation}
\label{eq:geo-val}
\xi_{lin}(S) = \begin{cases}
\nu_{lin}(S-\{g\}) & \text{if } g \in S \\
0 & \text{otherwise}
\end{cases}
\end{equation}
and
\begin{equation}
\label{eq:geo-val2}
\xi_{Met}(S) = \begin{cases}
\nu_{Met}(S-\{g\}) & \text{if } g \in S \\
0 & \text{otherwise}.
\end{cases}
\end{equation}
where $\nu_{lin}$ and $\nu_{Met}$ are defined in \eqref{eq:geo-linmet}.
\begin{thm}
In a geographic crowd-sourced system with value function $\xi$ defined in \eqref{eq:geo-val},\eqref{eq:geo-val2}, the Shapley allocation to the founder and agent $i$ for each $i \in M$ is given by
\begin{equation}
    \phi_{g}(\xi_{lin}) = \rho \frac{\sum_i n_i}{2} , \quad \phi_{i}(\xi_{lin}) = \rho \frac{n_i}{2} \\
\end{equation}
\begin{equation}
\begin{aligned}
\phi_{g}(\xi_{Met}) &= \rho \left(\frac{\left(\sum_i n_i\right)^2}{3} + \frac{\sum_i n_i^2}{6}\right) \\ 
\phi_{i}(\xi_{Met}) &= \rho \left(\frac{2(\sum_i n_i) n_i}{3} - \frac{n_i^2}{6}\right).
\end{aligned}
\end{equation}
\end{thm}
\begin{proof}
Follows from theorem \ref{th:smmuu-res}.
\end{proof}

As in Remark \ref{rem:alpha-fair}, notice that $\frac{\phi_g(\xi_{Met})}{\xi_{Met}(M)}$ and $\frac{\sum_i \phi_i(\xi_{Met})}{\xi_{Met}(M)}$ asymptotically approach $\frac{1}{3}$ and $\frac{2}{3}$ as the number of agents increases.

\section{Empirical Analysis}
Our main thesis in this paper is that many crowd-sourced systems, old or new, do not compensate their crowd participants and when they do so, the share is a relatively small fraction of their total revenue.  Our goal is to determine what a reasonable share would be if fair allocation existed.  Unfortunately, the information about sharing of the total value/revenue with crowd participants, if and when this information exists, is hard to come by.  Often data or content from the crowd is collected as part of the service level agreement signed by the crowd participants in exchange for a service, which may be search, access to social media platform, use of navigational maps, access to wireless spectrum and others services.  In public records, it is often hard to carve out compensation paid to crowd participants directly.  In a few instances, however, it is possible to estimate these payouts.  Here, we give 3 examples: Spotify, Pandora and YouTube:
\begin{itemize}

\item Spotify compensates audio contributors an average of 0.003-0.005 cents per stream which comes to about 2/3 of its total revenues in 2021 (\$11.4B), see \cite{TwoStory} and \cite{USAToday}.

\item Pandora pays an average of 0.00133 cents per stream which translates to 70\% of its total revenues to artists and publishers in 2021 (total revenue \$3.64 B in 2021)) of which 57\% goes to artists or record labels, ie 2/5 of the total revenue, see \cite{TimesInt}.

\item YouTube is our 3rd example and here it is possible to get more detailed information from public records which we describe below.

\end{itemize}

YouTube, an Alphabet subsidiary, has for over a decade awarded its content creators a large share of its advertising revenue.  Since 2019, YouTube has disclosed its payments to its crowd content creators.  Table \ref{tab:rev} below shows the breakout of the total revenues of Alphabet and YouTube in 2017 through 2021. Table \ref{tab:rev} gives us the YouTube annual revenues for 5 recent years.  But how much of these were paid to content creators in each fiscal year?  During August 2021, articles in Fortune \cite{Fortune}, Barron's and numerous other outlets reported that over the previous three years Alphabet/YouTube had paid \$30B to its content creators.  Adding the relevant figures from Table \ref{tab:rev}, this corresponds to a fraction 30/(28.8/2+19.7+15.1+11.1/2) = 30/54.75 = 55\% share of the total revenues during 2nd half of 2018, 2019, 2020, and 1st half of 2021 to content creators.  A more recent article \cite{Screen} puts this ratio at 68\%.  Other figures in the same range are also mentioned in various articles in the past 2 years.  
Putting all of this together, we observe a range of 55\%-68\% of the total YouTube revenues were paid out to content creators in each fiscal year, which is in close agreement with the $\left[\frac{1}{2},\frac{2}{3}\right]$ range we derived under the Shapley allocation using the simple value functions discussed above.  


\begin{table}[h]
\begin{center}
\begin{tabular}{ |c|c|c| } 
\hline
Year & Revenue (\$B) & Revenue (\$B) \\
    & Alphabet       & YouTube \\
\hline
2017 & 110.8 & 8.1 \\ 
\hline
2018 & 136.9 & 11.1 \\ 
\hline
2019 & 161.8 & 15.1 \\ 
\hline
2020 & 182.5 & 19.7 \\ 
\hline
2021 & 257.6 & 28.8 \\ 
\hline
\end{tabular}
\end{center}
\caption{Annual revenues of Alphabet and YouTube, 2017-2021}
\label{tab:rev}
\end{table}

Other video content crowd-sourced systems, such as TikTok, do not publish corresponding data although it is believed they also compensate their content creators a substantial share of their revenues.  Meta/Facebook has also started compensating its content creators\footnote{https://www.facebook.com/business/learn/lessons/how-make-money-facebook.}.

\section{Conclusions and Future Work}
In this paper, we defined value functions which succinctly capture the interplay between a small number of infrastructure providers, founders, and a large number of contributors, crowd participants, in cooperative value-generating settings.  We derived closed form expressions for fair allocation of the total value of such crowed-sourced systems between the founders and the crowd participants using Shapley allocation. We showed that Shapley allocation gives the crowd participants one-half to two-thirds share of the total payoff for the cases of linear and quadratic growth of revenue as a function of crowd size. We further showed similar results hold in cases of oligopolies of such network systems characterized by bilateral agreements between pairs of independent crowd-sourced systems. The same holds true for geographic crowd-sourced systems. These results are in contrast to the practice of many existing crowd-sourced systems whereby participants are awarded tokens of limited value compared to the global payoff of the full system but in close agreement with some others which we documented.  A natural next step would be to consider integration of our proposed value-based allocation in the emerging token-based reward systems which are being actively pursued by standards bodies, including the Web3.0 foundation.
\\
\\
\bibliographystyle{alpha}
\bibliography{references}

\newcommand{\etalchar}[1]{$^{#1}$}
\begin{thebibliography}{MCLM10}

\bibitem[Ale22]{TwoStory}
Marco Alexis.
\newblock How much does spotify pay per stream?
\newblock \url{https://twostorymelody.com/spotify-pay-per-stream/}, 2022.

\bibitem[BS21]{Fortune}
Mark Bergen and Lucas Shaw.
\newblock Youtube has paid out \$30 billion to creators as the competition for
  online content intensifies?
\newblock
  \url{https://fortune.com/2021/08/23/youtube-30-billion-ad-sales-online-creators/},
  2021.

\bibitem[CCH08]{Cheung08}
Yang Cheung, Dah~Ming Chiu, and Jianwei Huang.
\newblock Can bilateral {ISP} peering lead to network-wide cooperative
  settlement.
\newblock In {\em 2008 Proceedings of 17th International Conference on Computer
  Communications and Networks}, pages 1--6. IEEE, 2008.

\bibitem[Dea09]{dean09}
Jeffrey Dean.
\newblock Challenges in building large-scale information retrieval systems.
\newblock In {\em WSDM '09: Proceedings of the Second ACM International
  Conference on Web Search and Data}, pages 1--11. ACM, 2009.

\bibitem[Hel]{Helium}
Helium, people-powered network.
\newblock \url{http://helium.com}.

\bibitem[Kor20]{TimesInt}
Robert Kormoczi.
\newblock How much can an artist earn with pandora radio?
\newblock
  \url{https://timesinternational.net/how-much-can-an-artist-earn-with-pandora-radio/},
  2020.

\bibitem[Lan13]{Lan13}
Jaron Lanier.
\newblock {\em Who owns the future?}
\newblock Simon and Shuster, 2013.

\bibitem[MCL{\etalchar{+}}10]{MaRu}
Richard~TB Ma, Dah~Ming Chiu, John~CS Lui, Vishal Misra, and Dan Rubenstein.
\newblock On cooperative settlement between content, transit, and eyeball
  internet service providers.
\newblock {\em IEEE/ACM Transactions on networking}, 19(3):802--815, 2010.

\bibitem[MCLM10]{Ma10}
Richard~TB Ma, Dah~Ming Chiu, John~CS Lui, and Vishal Misra.
\newblock Internet economics: The use of {Shapley} value for {ISP} settlement.
\newblock {\em IEEE/ACM Transactions on networking}, 18(3):775--787, 2010.

\bibitem[Met13]{MC40}
Bob Metcalfe.
\newblock Metcalfe's law after 40 years of ethernet.
\newblock {\em Computer}, 46(12):26--31, 2013.

\bibitem[MICM10]{Mis10}
Vishal Misra, Stratis Ioannidis, Augustin Chaintreau, and Laurent
  Massouli{\'e}.
\newblock Incentivizing peer-assisted services: A fluid {Shapley} value
  approach.
\newblock {\em ACM SIGMETRICS Performance Evaluation Review}, 38(1):215--226,
  2010.

\bibitem[Mul22]{USAToday}
Clare Mulroy.
\newblock Spotify pays artists (sort of), but not per stream.
\newblock
  \url{https://www.usatoday.com/story/life/2022/10/22/how-much-per-spotify-stream/8094437001/},
  2022.

\bibitem[Pet18]{bitcoinMC}
Timothy Peterson.
\newblock Metcalfe's law as a model for bitcoin's value.
\newblock {\em Alternative Investment Analyst Review Q}, 2, 2018.

\bibitem[Rot88]{roth}
Alvin~E Roth.
\newblock Introduction to the {Shapley} value.
\newblock {\em The Shapley value}, pages 1--27, 1988.

\bibitem[Scr21]{Screen}
How much does youtube pay video creators?
\newblock
  \url{https://screencast-o-matic.com/blog/how-much-does-youtube-pay-video-creators},
  2021.

\bibitem[Sha51]{shapley}
Lloyd~S Shapley.
\newblock Notes on the n-person game {-- II}: The value of an n-person game.
\newblock {\em RAND Corporation}, 1951.

\bibitem[Vos19]{sv19}
Sherwin Voshmgir.
\newblock {\em Token Economy: How the Web3 reinvents the Internet}.
\newblock Token Kitchen, 2019.

\bibitem[Waz]{Waze}
Waze, company.
\newblock \url{http://waze.com/company}.

\bibitem[Web]{Web3}
Web 3.0 {F}oundation.
\newblock \url{https://web3.foundation/}.

\bibitem[ZLX15]{socialMC}
Xing-Zhou Zhang, Jing-Jie Liu, and Zhi-Wei Xu.
\newblock Tencent and facebook data validate metcalfe’s law.
\newblock {\em Journal of Computer Science and Technology}, 30(2):246--251,
  2015.

\end{thebibliography}

\pagebreak
\section{Appendices}
\subsection{Proofs from Section 3}
\label{app3}

\begin{proof}[Proof of Lemma \ref{lem:first}]
We can compute the Shapley value of $g$ as
\[
    \phi_g(\nu) = \frac{1}{|N|!}\sum\limits_{S \subset  N-\{g\}}|S|! (|N|-|S|-1)!\nu(S\cup\{g\})
\]
since $\nu(S) = 0$ as $g \notin S$.
 Since $\nu(S\cup\{g\})$ is just a function of $|S|$, we can reduce the computation as
 
 \begin{align*}
    \phi_g(\nu) &= \frac{1}{|N|!}\sum\limits_{S \subset  N-\{g\}}|S|! (|N|-|S|-1)!f(|S|) \\
    &= \frac{1}{(n+1)!} \sum_{s=0}^{n} {n\choose s} s! (n-s)!f(s) \\
    &= \frac{1}{n+1} \sum_{s=0}^{n} f(s).
 \end{align*}
Since the participants are all identical, we derive the Shapley value as
 
 \begin{align*}
     \phi_{u_i}(\nu) = \frac{\nu(N)-\phi_g(\nu)}{n}.
 \end{align*}
 
\end{proof}

\begin{proof}[Proof of Theorem \ref{th:smmuu-res}]
The Shapley value of participant $i$ can be computed as
\begin{align*}
    \phi_{u_i}(\nu) &= \sum_{S \subset N-\{u_i\}} \frac{|S|!(|N|-|S|-1)!}{|N|!} \left(\nu(S\cup\{u_i\}) - \nu(S) \right) \\
    &= \sum_{S\subset N-\{u_i\}, g \in S} \frac{|S|!(|N|-|S|-1)!}{|N|!}\left(\nu(S\cup\{u_i\})-\nu(S)\right).
\end{align*}

The marginal utility of participant i can be written as
\begin{align*}
    \nu(S\cup\{u_i\}) - \nu(S) = W_i^{2\alpha} + \sum_{j \in S} 2W_i^{\alpha}W_j^{\alpha},
\end{align*}
which gives
\begin{align*}
    \phi_{u_i}(\nu) =& \sum_{S \subset N-\{u_i\}, g \in S} \frac{|S|!(|N|-|S|-1)!}{|N|!} \left(W_i^{2\alpha} + \sum_{j \in S} 2W_i^{\alpha}W_j^{\alpha}\right) \\
    =& \frac{1}{(n+1)!}\sum_{s=1}^{n} s!(n-s)! {n-1\choose s-1} W_i^{2\alpha} \ + \\
    & \sum_{S\subset N-\{u_i\}, g \in S} \frac{|S|!(|N|-|S|-1)!}{|N|!} \left( \sum_{j \in S} 2W_i^{\alpha}W_j^{\alpha}\right).
\end{align*}
The first sum can be simplified directly and the second sum maybe rewritten using Fubini's theorem as follows
\begin{align*}
    \phi_{u_i}(\nu) &= W_i^{2\alpha}\sum_{s=1}^n \frac{s}{n(n+1)} +  \sum_{j \neq i}\frac{2W_i^{\alpha}W_j^{\alpha}}{(n+1)!} \sum\limits_{\{g, u_j\} \subset S \subset N-\{u_i\}} |S|!(|N|-|S|-1)! \\
    &= \frac{W_i^{2\alpha}}{2} + \sum_{j \neq i} 2W_i^{\alpha}W_j^{\alpha}\sum_{s=2}^n \frac{s!(n-s)!}{(n+1)!}{n-2 \choose s-2} \\
    &= \frac{W_i^{2\alpha}}{2} + \sum_{j \neq i} 2W_i^{\alpha}W_j^{\alpha} \sum_{s=2}^n\frac{s(s-1)}{(n+1)(n-1)n} \\
    &= \frac{W_i^{2\alpha}}{2} + \sum_{j\neq i} \left(\frac{2}{3}+O(n^{-1})\right) W_i^{\alpha}W_j^{\alpha}.
\end{align*}
Dividing by the value of the grand coalition $\nu(N)$, we get
\begin{align*}
    \frac{\phi_{u_i}(\nu)}{\nu(N)} &= \frac{1}{2}\frac{W_i^{2\alpha}}{\left(\sum_j W_j^\alpha \right)^2} + \frac{2}{3} \sum_{j \neq i}
    \frac{W_i^{\alpha}W_j^{\alpha}}{\left(\sum_k W_k^\alpha \right)^2} + O(n^{-1})\sum_{j \neq i} \frac{W_i^{\alpha}W_j^{\alpha}}{\left(\sum_k W_k^\alpha \right)^2} \\
    &= \frac{f_i^2}{2} + \frac{2}{3}f_i\sum_{j\neq i}f_j   + O(n^{-1}) = \frac{f_i^2}{2} + \frac{2}{3}f_i(1-f_i) + O(n^{-1}) \\
    &= \frac{2}{3}f_i - \frac{f_i^2}{6} + O(n^{-1}).
\end{align*}
It follows that
\begin{align*}
\lim_{n \to \infty} \frac{\phi_{u_i}(\nu)}{\nu(N)} &= \frac{2}{3}f_i - \frac{f_i^2}{6}, \\
    \lim_{n \to \infty} \frac{\phi_g(\nu)}{\nu(N)} &= 1-\frac{2}{3}\sum_{i}f_i + \sum_i \frac{f_i^2}{6} = \frac{1}{3} + \sum_i \frac{f_i^2}{6}. 
\end{align*}
\end{proof}

\subsection{Proofs from Section 4}
\label{app4}

\begin{proof}[Proof of Theorem \ref{th:network-shap-res}]
 The marginal utility of $u$ to coalition $S$ is
 \begin{align*}
     \nu(S \cup \{u\}) - \nu(S) =& n_u^2 + \sum_{(u, w) \in E, w \in S} n_un_w + \sum_{(w, u) \in E, w \in S} n_un_w \\
     =&  n_u^2 + \sum_{(u, w) \in E, w \in S} 2n_un_w.
 \end{align*}
 We can then compute
 \begin{align*}
    \phi_u(\nu) =& \sum_{S \subset V-\{u\}} \frac{|S|! (|V| - |S| -1)!}{|V|!} \left(\nu(S \cup \{u\}) - \nu(S) \right) \\
    =& \sum_{S \subset V-\{u\}} \frac{|S|! (|V| - |S| -1)!}{|V|!}n_u^2 +\\
    & \sum_{S \subset V-\{u\}} \frac{|S|! (|V| - |S| -1)!}{|V|!} \sum_{(u, w) \in E, w \in S} 2n_un_w  \\
    =& n_u^2 + \sum_{S \subset V-\{u\}} \frac{|S|! (|V| - |S| -1)!}{|V|!} \sum_{(u, w) \in E, w \in S} 2n_un_w.
 \end{align*}
 The second term above can be reduced using Fubini's theorem as we did before to get
 \begin{align*}
    \phi_u(\nu) &= n_u^2 +  \sum_{(u, w) \in E}\sum_{w \in S \subset V - \{u\}} \frac{|S|! (|V| - |S| -1)!}{|V|!}  2n_un_w \\
    &= n_u^2 + \sum_{(u,w) \in E} 2n_u n_w \sum_{w \in S \subset V - \{u\}}\frac{|S|! (|V| - |S| -1)!}{|V|!} \\
    &= n_u^2 + \sum_{(u,w) \in E} 2n_u n_w \sum_{s = 1}^{|V|-1} \frac{s!(|V|-s-1)!}{|V|!} {|V|-2 \choose s-1} \\
    &= n_u^2 + \sum_{(u,w) \in E}2n_u n_w \sum_{s=1}^{|V|-1} \frac{s}{|V|(|V|-1)}\\
    &= n_u^2 + \sum_{w \in N(u)}n_un_w.
 \end{align*}
 
\end{proof}

\begin{proof}[Proof of Theorem \ref{th:network-shap-res1}]
The value function $\xi$ can be split as
\begin{equation}
    \xi(S) = \sum_{v \in V} \xi_u(S) + \sum_{(v,w) \in E} \xi_{u,w}(S)
\end{equation}
where
\[
    \xi_v(S) = \begin{cases}
|S_v|^2 & \text{if } v \in S \\
0 & \text{otherwise},
\end{cases}
\]
and
\[
\xi_{v,w}(S) = \begin{cases}
2|S_v||S_w| & \text{if } v, w \in S \\
0 & \text{otherwise}.
\end{cases}
\]
Each of these value functions can be considered as the value function of a cooperative game in its own right. By Linearity of Shapley value,
\begin{equation}
\label{eq:shapleysum}
    \phi_{a}(\xi) = \sum_{v \in V} \phi_a(\xi_v) + \sum_{(v,w) \in E} \phi_a(\xi_{v,w})
\end{equation}
for any player $a \in N$.

Notice that for any $v \in V$, if $a \notin \{v\} \cup U_v$ then
\[
    \xi_v(S \cup \{a\}) - \xi_v(S) = 0, \quad \text{for all $S \subset N$},
\]
which means $a$ is a null player in the game $\xi_v$. This implies $\phi_a(\xi_v) = 0$ by the null player axiom of Shapley value. The game $\xi_v$ on the remaining set $\{v\}\cup U_v$ is just the single major player game \ref{eq:smmu-model}, and the result of Theorem \ref{th:smmu-res} gives us
\[
    \phi_v(\xi_v) = \frac{n_v^2}{3} + O(n_v), \quad \phi_{u_v}(\xi_v) = \frac{2n_v}{3} - O(1).
\]

Similarly, for any $a \notin  \{v,w\} \cup U_v \cup U_w$
\[
    \xi_{v,w}(S \cup \{a\}) - \xi_{v,w}(S) = 0, \quad \text{for all $S \subset N$},
\]
which implies $\phi_a(\xi_{v,w}) = 0$ by the null axiom of Shapley value. 

All that remains to be computed is the Shapley value of players $\{v,w\} \cup U_v \cup U_w$ with the value function $\xi_{v,w}$. Since we have shown that all other players in this game are dummy, we can find the Shapley value of players $\{v,w\} \cup U_v \cup U_w$ as their Shapley value in the simplified game whose set of players is $\tilde{N} = \{v,w\} \cup U_v \cup U_w$ and value function $\xi_{v,w}$. To do this, we again split the value function $\xi_{v,w}$ as
\[
    \xi_{v,w}(S) = \sum_{u_v \in U_v} \sum_{u_w \in U_w} \xi_{v,w}^{u_v, u_w}(S)
\]
where
\[
    \xi_{v,w}^{u_v, u_w}(S) = \begin{cases}
2 & \text{if } S = \{v, w, u_v, u_w\} \\
0 & \text{otherwise},
\end{cases}
\]
and compute the Shapley value of the players for each of these value functions separately.
Observe that any player $\bar{u_v} \in U_v-\{u_v\}$ or $\bar{u_w} \in U_w-\{u_w\}$ is a null player of the value function $\xi_{v,w}^{u_v, u_w}$, which means
\[
    \phi_{a}(\xi_{v,w}^{u_v,u_w}) = 0, \quad \text{if } U_v \ni a \neq u_v \text{ or } U_w \ni a \neq u_w.
\]

This means we can again reduce the set of players to $\{v, w, u_v, u_w\}$ and value function to $\xi_{v,w}^{u_v,u_w}$, and notice that this value function is symmetric with respect to all four players $\{v, w, u_v, u_w\}$ which implies
\[
    \phi_{a}(\xi_{v,w}^{u_v,u_w}) = \begin{cases}
\frac{1}{2} & \text{if } a \in \{v, w, u_v, u_w\} \\
0 & \text{otherwise}.
\end{cases}
\]
We can now compute the original Shapley values as
\begin{align*}
    \phi_{a}(\xi_{v,w}) = \sum_{u_v \in U_v}\sum_{u_w \in U_w} \phi_a(\xi_{v,w}^{u_v,u_w}).
\end{align*}
When $a = v$ or $a = w$,
\begin{align*}
    \phi_{a}(\xi_{v,w}) = \sum_{u_v \in U_v}\sum_{u_w \in U_w} \frac{1}{2} = \frac{|U_v||U_w|}{2} = \frac{n_vn_w}{2}.
\end{align*}
When $a \in U_v$,
\begin{align*}
    \phi_{a}(\xi_{v,w}) = \sum_{u_v \in U_v}\sum_{u_w \in U_w} \frac{1}{2} \delta_{u_v, a} = \frac{|U_w|}{2} = \frac{n_w}{2}
\end{align*}
and when $a \in U_w$,
\begin{align*}
    \phi_{a}(\xi_{v,w}) = \sum_{u_v \in U_v}\sum_{u_w \in U_w} \frac{1}{2} \delta_{u_w, a} = \frac{|U_v|}{2} = \frac{n_v}{2}.
\end{align*}
Putting everything together in \eqref{eq:shapleysum}, we get the results of equation \eqref{eq:network-shap-res} and \eqref{eq:network-shap-res-minor}. 
\end{proof}

\end{document}